\documentclass[11pt]{article}
\usepackage{amsmath,amsthm,amssymb,epsfig,sectsty,caption,color}
\usepackage[letterpaper,hmargin=0.95in,vmargin=1.1in]{geometry}

\usepackage{algorithm, algpseudocode}
\algrenewcommand\algorithmicrequire{\textbf{\quad Input:}}
\algrenewcommand\algorithmicensure{\textbf{\quad Output:}}

\theoremstyle{plain}
\newtheorem{theorem}{Theorem}[section]
\newtheorem{lemma}[theorem]{Lemma}
\newtheorem{claim}[theorem]{Claim}

\theoremstyle{definition}

\newcommand{\qedsymb}{\hfill{\rule{2mm}{2mm}}}
\renewenvironment{proof}{\begin{trivlist} \item[\hspace{\labelsep}{\bf \noindent Proof.\/}] }{\qedsymb\end{trivlist}}%
\newcommand{\bbR}{\mathbb{R}}
\newcommand{\bbN}{\mathbb{N}}
\newcommand{\floor}[1]{\lfloor #1 \rfloor}
\newcommand{\opt}{\mathrm{OPT}}
\newcommand{\alg}{\mathrm{ALG}}

\sectionfont{\large} \subsectionfont{\normalsize}

\begin{document}
\begin{titlepage}
\title{Ranking with Submodular Valuations}
\author{
Yossi Azar\thanks{Blavatnik School of Computer Science, Tel-Aviv
University, Tel-Aviv 69978, Israel. Email: {\tt azar@tau.ac.il}.}
    \and
Iftah Gamzu\thanks{Blavatnik School of Computer Science, Tel-Aviv
University, Tel-Aviv 69978, Israel. Email: {\tt iftgam@tau.ac.il}.
Supported by the Israel Science Foundation, by the European
Commission under the Integrated Project QAP funded by the IST
directorate as Contract Number 015848, by a European Research
Council (ERC) Starting Grant, and by the Wolfson Family Charitable
Trust.} }
\date{}
\maketitle

\begin{abstract}

We study the problem of ranking with submodular valuations. An
instance of this problem consists of a ground set $[m]$, and a
collection of $n$ monotone submodular set functions $f^1, \ldots,
f^n$, where each $f^i: 2^{[m]} \to \bbR_+$. An additional
ingredient of the input is a weight vector $w \in \bbR_+^n$. The
objective is to find a linear ordering of the ground set elements
that minimizes the weighted cover time of the functions. The cover
time of a function is the minimal number of elements in the prefix
of the linear ordering that form a set whose corresponding
function value is greater than a unit threshold value.

Our main contribution is an $O(\ln(1 / \epsilon))$-approximation
algorithm for the problem, where $\epsilon$ is the smallest
non-zero marginal value that any function may gain from some
element. Our algorithm orders the elements using an adaptive
residual updates scheme, which may be of independent interest. We
also prove that the problem is $\Omega(\ln(1 / \epsilon))$-hard to
approximate, unless $\mathrm{P} = \mathrm{NP}$. This implies that
the outcome of our algorithm is optimal up to constant factors.
\end{abstract}

\thispagestyle{empty}
\end{titlepage}

\section{Introduction} \label{sec:Intro}

Let $f: 2^{[m]} \rightarrow \bbR$ be a set function, where $[m] =
\{1, 2, \ldots, m\}$. The function $f$ is \emph{submodular} iff
$$
f(S) + f(T) \geq f(S \cup T) + f(S \cap T) \ ,
$$
for all $S, T \subseteq [m]$. An alternative definition of
submodularity is through the property of decreasing marginal
values. Given a function $f: 2^{[m]} \rightarrow \bbR$ and a set
$S \subseteq [m]$, the function $f_S$ is defined by $f_S(j) = f(S
\cup \{j\}) - f(S)$. The value $f_S(j)$ is called the incremental
marginal value of element $j$ to the set $S$. The \emph{decreasing
marginal values} property requires that $f_S(j)$ is non-increasing
function of $S$ for every fixed $j$. Formally, it requires that
$f_S(j) \geq f_T(j)$, for all $S \subseteq T$ and $j \in [m]
\setminus T$. Since the amount of information necessary to convey
an arbitrary submodular function may be exponential, we assume a
value oracle access to the function. A \emph{value oracle} for $f$
allows us to query about the value of $f(S)$ for any set $S$.
Throughout the rest of the paper, whenever we refer to submodular
functions, we shall also imply \emph{normalized} and
\emph{monotone} functions. Specifically, we assume that a
submodular function $f$ also satisfies $f(\emptyset) = 0$ and
$f(S) \leq f(T)$ whenever $S \subseteq T$.

In this paper, we focus our attention on the problem of ranking
with submodular valuations. An instance of this problem consists
of a ground set $[m]$, and a collection of $n$ monotone submodular
set functions $f^1, \ldots, f^n$, where each $f^i: 2^{[m]} \to
\bbR_+$. An additional ingredient of the input is a weight vector
$w \in \bbR_+^n$. The objective is to find a linear ordering of
the ground set elements that minimizes the weighted cover time of
the functions. The cover time of a function is the minimal number
of elements in the prefix of the linear ordering that form a set
whose corresponding function value is greater than some
predetermined threshold. More precisely, the objective is to find
a linear ordering $\pi: [m] \to [m]$ that minimizes $\sum_{i =
1}^{n} w_i c_i$, where $c_i$ is the cover time of function $f^i$,
defined as the minimal index for which $f^i( \{ \pi(1), \ldots,
\pi(c_i) \}) \geq 1$. Here, $\pi(t)$ stands for the element
scheduled at time $t$ according to the linear ordering $\pi$. It
is worth noting that the fact that each cover time is defined with
respect to a unit threshold value does not limit the generality of
the problem. In particular, given a vector $\lambda \in \bbR_+^n$,
where $\lambda_i$ determines the cover threshold of function
$f^i$, one can obtain an equivalent instance by normalizing each
$f^i$ with $\lambda_i$, and updating all thresholds to $1$.

\subsection{Our results}

Our main contribution is an $O(\ln(1 / \epsilon))$-approximation
algorithm for ranking with submodular valuations, where $\epsilon
= \min \{f^i_S(j) > 0\}$ is the smallest non-zero marginal value
that any function may gain from some element. We note that
elements can have a marginal value of zero. Our algorithm orders
the ground set elements using an \emph{adaptive residual updates
scheme}, which iteratively selects an element that has a maximal
marginal contribution with respect to an appropriately defined
residual cover of the functions. This approach has similarities
with the well-known multiplicative weights method (see, e.g.,
\cite{PlotkinST95,GargK07}). Our algorithm is motivated by the
observation that the natural greedy algorithm, which iteratively
selects an element based on its absolute marginal contribution to
the cover of the functions, performs poorly. In particular, a
greedy type algorithm misjudges elements with low marginal
contribution as less important, and therefore, unwisely schedules
them late.

We also establish that ranking with submodular valuations is
$\Omega(\ln(1 / \epsilon))$-hard to approximate, assuming that
$\mathrm{P} \neq \mathrm{NP}$. This implies that the outcome of
our algorithm is optimal up to constant factors. This result is
attained by demonstrating that the restricted setting of our
problem in which there is a single function to cover already
incorporates the set cover problem as a special instance. We would
like to emphasize that even though this single function setting
captures the computational hardness of the problem, it does not
capture its algorithmic essence. The main algorithmic challenge
that is addressed by our scheme is to obtain a good performance
guarantee when there are many functions, each of which has a
different linear order that best suits its needs. In particular,
one can easily validate that the natural greedy algorithm is
essentially optimal in the single function setting. One additional
interesting consequence of this result is that our problem
generalizes both the set cover problem and its min-sum variant.
This is the first problem formulation that has this property.

\subsection{Applications} \label{subsec:Applications}

\noindent {\bf Web search ranking.} One impetus for studying the
above problem is an application in web search ranking. Web search
has become an important part in the daily lives of many people.
Recently, there has been a great interest in incorporating users
behavior into web search ranking. Essentially, these studies make
an effort to personalize the web search results (see, e.g.,
\cite{TeevanDH05,AgichteinBD06,DouSW07,DupretP08}). However, in
the absence of any explicit knowledge of user intent, one has to
focus on how to produce a set of diversified results that properly
account for the interests of the overall user
population~\cite{ClarkeKCVABM08,AgrawalGHI09}. In particular, it
seems natural to utilize logs of previous search sessions, and try
to minimize the average effort of all users in finding the web
pages that satisfy their needs. When performing web search, a user
usually reads the result items from top to
bottom~\cite{JoachimsGPHRG07}. The time a user spends on reading
the result items is the overhead in web search.

The problem of ranking with submodular valuations can model the
above-mentioned scenario as follows: there is a set of $m$ search
result items and there are $n$ user types of known proportion.
Each user type has a submodular relevance function that quantifies
the information that the user type gains from inspecting any
subset of result items. The goal is to order the result items in a
way that minimizes the average effort of the user types. The
effort of a user type is the number of result items it has to
review until it gains a critical mass of relevant information.
Notice that submodularity suits naturally for the ranking
application since information in result items overlaps and does
not necessarily complement each other.

\smallskip \noindent {\bf Broadcast in mobile networks.} Another
application that can be modeled by ranking with submodular
valuations is broadcast in mobile networks. In this scenario,
there is a base station that needs to sequentially transmit a set
of $m$ data segments. In addition, there is a collection of $n$
clients, each of which is interested in some individual target
data. Each data segment contains a mix of information that may be
relevant to a number of clients. The amount of information depends
both on the data segment and the client. Moreover, there is
informational redundancy between different data segments. This
allows clients to extract their relevant target data from
different subsets of segments. A client can extract her target
data once she receives sufficient relevant information from the
data segments. The goal is to set an order for the transmission of
the data segments that minimizes the average latency of the
clients. The latency of a client is the earliest time in which she
receives data segments that contain enough information to decode
her target data. Notice that the amount of relevant information
that each client extracts from the data segments is submodular.

\subsection{Previous work on special cases}

The problem of ranking with submodular valuations extends the
\emph{multiple intents ranking} problem~\cite{AzarGY09}. One can
demonstrate that an input instance for the latter problem can be
translated to an instance of ranking with submodular valuations in
which each function $f^i$ is linear, and the value that the
function has for any element is either $0$ or some value $\nu_i
\in (0,1]$ common to that function. The multiple intents ranking
problem is known to admit a constant approximation by the work of
Bansal, Gupta and Krishnaswamy~\cite{BansalGK10}. Specifically,
they presented a clever randomized LP rounding algorithm that
improved upon a previous logarithmic
approximation~\cite{AzarGY09}. The \emph{min-sum set cover}
problem can be modelled as a special case of multiple intents
ranking in which each $\nu_i$ is boolean, i.e., $\nu_i \in
\{0,1\}$. The best known result for this problem is a
$4$-approximation algorithm that was developed by Feige,
Lov{\'a}sz and Tetali~\cite{FeigeLT04}. This algorithm was
implicit in the work of Bay-Noy et al.~\cite{Bar-NoyBHST98}. The
former paper also proved that $4$-approximation is best possible,
unless $\mathrm{P} = \mathrm{NP}$. The \emph{minimum latency set
cover} problem is another special case of multiple intents ranking
in which each function $f^i$ has exactly $1 / \nu_i \in \bbN_+$
elements with non-negative value $\nu_i$. Hassin and
Levin~\cite{HassinL05} studied this problem, and observed that it
can be modeled as a special case of the classic
precedence-constrained scheduling problem $1|prec|\sum_j w_j C_j$.
The latter problem has various $2$-approximation algorithms (see,
e.g., the survey~\cite{ChekuriK04}). Woeginger~\cite{Woeginger03}
demonstrated that the special case derived from minimum latency
set cover is as hard to approximate as the general scheduling
problem. This implies, in conjunction with a recent work of Bansal
and Khot~\cite{BansalK09}, that it is hard to approximate minimum
latency set cover to within a factor better than $2$, assuming a
variant of the Unique Games Conjuncture.

\subsection{Other related work}
Submodular functions arise naturally in operations research and
combinatorial optimization. One of the most extensively studied
questions is how to minimize a submodular function. A series of
results demonstrate that this task can be performed efficiently,
either by the ellipsoid algorithm~\cite{GrotschelLS81} or through
strongly polynomial time combinatorial
algorithms~\cite{Schrijver00,IwataFF01,Iwata03,Orlin07,Iwata08,IwataO09}.
Recently, there has been a surge of interest in understanding the
limits of tractability of minimization problems in which the
classic linear objective function was replaced by a submodular one
(see, e.g.,~\cite{SvitkinaF08,GoemansHIM09,GoelKTW09,IwataN09}).
Notably, these submodular problems are commonly considerably
harder to approximate than their linear counterparts. For example,
the minimum spanning tree problem, which is polynomial time
solvable with linear cost functions is $\Omega(n)$-hard to
approximate with submodular cost functions~\cite{GoelKTW09}, and
the sparsest cut problem, which admits an $O(\sqrt{\ln
n})$-approximation algorithm when the cost is
linear~\cite{AroraHK10} becomes $\Omega(\sqrt{n / \ln n})$-hard to
approximate with submodular costs~\cite{SvitkinaF08}. Our work
extends the tools and techniques in this line of research. In
particular, our results establish a computational separation of
logarithmic order between the submodular settings and the linear
setting, which admits a constant factor
approximation~\cite{AzarG10}.

\section{An Adaptive Residual Updates Scheme} \label{sec:AdaptiveScheme}

In this section, we develop a deterministic algorithm for the
problem under consideration that has an approximation guarantee of
$O(\ln(1 / \epsilon))$, where $\epsilon = \min \{f^i_S(j)
> 0\}$ is the smallest non-zero marginal value that any function
may gain from some element. An interesting feature of our
algorithm is that it generalizes several previous algorithmic
results. For example, if our algorithm is given a multiple intents
ranking instance then it behaves like the harmonic interpolation
algorithm of~\cite{AzarGY09}, if it is given a min-sum set cover
instance then it reduces to the constant approximation algorithm
of~\cite{FeigeLT04}, and if it is given a (submodular) set cover
instance then it acts like the well-known greedy algorithm for the
set cover problem~\cite{Wolsey82b}. Nonetheless, it is important
to emphasize that all these algorithms use fixed values in their
computation; in contrast, our algorithm employs dynamically
changing values.

\subsection{The algorithm} \label{subsec:Algorithm}

The adaptive residual updates algorithm, formally described below,
works in steps. In each step, the algorithm extends the linear
ordering with a non-selected element that maximizes the weighted
sum of its corresponding potential values. The potential value
$P_{ij}$ of element $j$ for function $f^i$ is initially equal to
the marginal value $f^i_\emptyset(j)$. As the algorithm
progresses, it is adaptively updated with respect to the selected
elements and the residual cover of $f^i$, as formally presented in
line~\ref{alg:AdaptiveDef}. Intuitively, this update fashion gives
more influence to values corresponding to functions whose cover
draw near their thresholds. We emphasize that this dynamic update
fashion is different than that of the exponential weights and
harmonic interpolation techniques. Also note that our adaptive
residual updates scheme is motivated by the observation that the
natural greedy algorithm, which orders elements based on their
absolute marginal contribution, fails to provide good
approximation. This insufficiency is exhibited in
Appendix~\ref{appsubsec:NonAdaptiveFail}.

\begin{algorithm}
\caption{Adaptive Residual Updates}\label{cap:AdaptiveUpdate}%
\begin{algorithmic}[1]
\Require A collection of $n$ submodular set functions $f^1, \ldots, f^n$, where each $f^i: 2^{[m]} \to \bbR_+$ %
\Statex \qquad\quad A weight vector $w \in \bbR_+^n$ %
\Ensure A linear ordering $\pi: [m] \to [m]$ of the ground set elements \smallskip %

\State $S \leftarrow \emptyset$ %
\For{$t \leftarrow 1$ to $m$} %
    \For{$i \leftarrow 1$ to $n$} %
        \ForAll{$j \in [m] \setminus S$} %
            \State $P_{ij} \leftarrow \begin{cases} 0 & f^i(S) \geq 1, \\ \min\left\{1, \frac{f^i(S \cup \{j\}) - f^i(S)}{1 - f^i(S)}\right\} & \makebox{otherwise} \end{cases}$ \label{alg:AdaptiveDef} %
        \EndFor %
    \EndFor %
    \State Let $j \in [m] \setminus S$ be the element with maximal $\sum_{i = 1}^{n}w_i P_{ij}$ %
    \State $S \leftarrow S \cup \{ j \}$ %
    \State $\pi(t) \leftarrow j$ %
\EndFor %
\end{algorithmic}
\end{algorithm}

\subsection{Analysis} \label{subsec:Analysis}

In the remainder of this section, we analyze the performance of
the algorithm. There are several techniques that we employ. We
begin by establishing an interesting algebraic inequality
applicable for any monotone function and any arbitrary sequence of
element additions. For the purpose of bounding the cost of the
solution of our algorithm, we compare it to a collection of
solutions induced by the optimal linear ordering applied to
truncated instances of the problem. We also utilize and extend the
analysis methods presented in~\cite{FeigeLT04,AzarGY09}.

\begin{theorem} \label{thm:Approx}
The adaptive residual updates algorithm constructs a linear
ordering whose induced cost is no more than $O(\ln(1 / \epsilon))$
times the optimal one.
\end{theorem}
\begin{proof}
We begin by introducing the notation and terminology to be used
throughout this proof:

\begin{itemize}
\item Let $\opt$ and $\alg$ be the cost induced by the optimal
linear ordering and the linear ordering $\pi$ constructed by the
algorithm, respectively.

\item Let $P$ be the final state of the potential values matrix
maintained during the algorithm. Notice that the potential values
of each column of $P$ are the values that the corresponding
element had when it was selected by the algorithm.

\item Let $I_t = \{ i : c_i \geq t \}$ be the indices of the
functions that were not covered before step $t$ of the algorithm,
and let $R_t$ be the relative cost of step $t$ of the algorithm.
Specifically, $R_t = \sum_{i \in I_t} w_i$. In addition, let $Q_t
= \sum_{i=1}^{n} w_i P_{i\pi(t)}$ be the weighted sum of potential
values corresponding to the element selected at step $t$ of the
algorithm, and $\Delta_{t'} = \sum_{t = t'}^{m} Q_t$. Finally, let
$\Lambda_t = R_t / Q_t$ be the penalty of step $t$.
\end{itemize}

We can now reinterpret the cost induced by the linear ordering
$\pi$ using the mentioned notation. In particular, one can
validate that $\alg = \sum_{i=1}^{n} w_i c_i  = \sum_{t=1}^{m} R_t
= \sum_{t=1}^{m} \Lambda_t Q_t$. The next lemma bounds
$\Delta_{t'}$ in terms of $R_{t'}$.

\begin{lemma} \label{lemma:Width}
Let $\gamma = \ln(1 / \epsilon) + 2$. Then, $\Delta_{t'} \leq
\gamma R_{t'}$, for every $t' \in [m]$.
\end{lemma}
\begin{proof}
In what follows, we demonstrate that $\sum_{t = 1}^m P_{i \pi(t)}
\leq \gamma$ holds for every function $f^i$. Notice that if we
establish this argument then the lemma follows since
$$
\Delta_{t'} = \sum_{t = t'}^m \sum_{i=1}^{n} w_i P_{i\pi(t)} =
\sum_{t = t'}^m \sum_{i \in I_{t'}} w_i P_{i\pi(t)} \leq \sum_{t =
1}^m \sum_{i \in I_{t'}} w_i P_{i\pi(t)} \leq \gamma \sum_{i \in
I_{t'}} w_i = \gamma R_{t'} \ ,
$$
where the second equality results by noticing that each function
$f^i$ covered before step $t'$ must have $P_{i\pi(t)} = 0$ for
every $t \geq t'$. As a result, the only functions that may have
strictly positive potential values are those that were not covered
before step $t'$, namely, those in $I_{t'}$.

Consider the function $f^i$, and let us assume that its cover time
is $k$. Let $S_t = \{ \pi(1), \ldots, \pi(t) \}$ be the set of
elements ordered up to (and including) step $t$ according to
$\pi$. In particular, let $S_0 = \emptyset$. One can verify that
the potential values of function $f^i$ satisfy $P_{i \pi(t)} =
(f^i(S_{t}) - f^i(S_{t-1}))/(1 - f^i(S_{t-1}))$ for every $t < k$,
$P_{i \pi(k)} = 1$, and $P_{i \pi(t)} = 0$ for every $t
> k$. Consequently, we get that
$$
\sum_{t=1}^{m} P_{i \pi(t)} = \sum_{t=1}^{k-1} \frac{f^i(S_{t}) -
f^i(S_{t-1})}{1 - f^i(S_{t-1})} + 1 \leq \ln(1 / \epsilon) + 2  \
,
$$
where the inequality follows from
Claim~\ref{claim:SequenceStretch}. This claim establishes a
generic bound which applies to any monotone function and any
arbitrary sequence of element additions. The desired bound is
obtained by utilizing the claim with respect to the submodular
function $f^i$ and the collection of sets $S_0, \ldots, S_{k-1}$.
One should also notice that $f^i(S_{k-1}) < 1$ by construction,
and that $\epsilon \leq \delta$.~
\end{proof}

\begin{claim} \label{claim:SequenceStretch}
Given a monotone function $f: 2^{[m]} \to \bbR_+$, and a
collection of set $S_0 \subseteq \cdots \subseteq S_\ell \subseteq
[m]$ such that $f(S_0) = 0$ and $f(S_\ell) < 1$ then
$$
\sum_{t = 1}^{\ell} \frac{f(S_{t}) - f(S_{t-1})}{1 - f(S_{t-1})}
\leq \ln(1 / \delta) + 1 \ ,
$$
where $\delta = \min_t\{f(S_{t}) - f(S_{t-1}) > 0\}$.
\end{claim}
\begin{proof}
The monotonicity property of the function $f$ guarantees that $0 =
f(S_0) \leq \cdots \leq f(S_\ell) \leq 1$. We can also assume
without loss of generality that $f(S_\ell) - f(S_{\ell-1}) \geq
\delta > 0$, since otherwise, the last term in the above-mentioned
summation must be equal to $0$, and therefore, may be neglected.
Now, notice that for any $t = 1, \ldots, \ell - 1$,
$$
\int_{f(S_{t-1})}^{f(S_{t})} \frac{1}{1-x}dx \geq
\int_{f(S_{t-1})}^{f(S_{t})} \frac{1}{1-f(S_{t-1})}dx =
\frac{f(S_{t}) - f(S_{t-1})}{1-f(S_{t-1})} \ ,
$$
where the first inequality results from the fact that the function
$1 / (1 - x)$ is monotonically increasing for $x \in [0,1)$.
Furthermore, notice that $(f(S_\ell) - f(S_{\ell-1})) / (1 -
f(S_{\ell-1})) < 1$. This simply follows since we know that
$f(S_\ell) < 1$. Combining the previously stated arguments, we
attain that
\begin{eqnarray*}
\sum_{t = 1}^{\ell} \frac{f(S_{t}) - f(S_{t-1})}{1 - f(S_{t-1})} &
\leq & \sum_{t = 1}^{\ell - 1} \frac{f(S_{t}) - f(S_{t-1})}{1 -
f(S_{t-1})} + 1 \leq \sum_{t = 1}^{\ell - 1} \int_{f(S_{t-1})}^{f(S_{t})} \frac{1}{1-x}dx + 1 \\
& = &\int_{f(S_0)}^{f(S_{\ell-1})} \frac{1}{1-x}dx + 1 =  \ln
\left(\frac{1}{1 - f(S_{\ell-1})}\right) + 1 \leq \ln(1 / \delta)
+ 1 \ ,
\end{eqnarray*}
where the last inequality holds since $1-f(S_{\ell-1}) > f(S_\ell)
-f(S_{\ell-1}) \geq \delta$.~
\end{proof}

We continue by introducing a collection of histograms. These
histograms will be utilized to bound the cost of the algorithm in
terms of the cost of the optimal solution.

\smallskip \noindent {\bf The optimal solution as a histogram.}
The histogram that relates to the optimal solution consists of $n$
bars, one for each function. The bar associated with function
$f^i$ has a width of $w_i$, while its height is equal to the step
in which that function was covered in the optimal solution. A
function is regarded as covered at step $t$ if its cover time
according to the optimal ordering is $t$. The bars in the
histogram are ordered according to the steps in which the
corresponding functions were covered in the optimal solution.
Notice that this implies that the histogram is non-decreasing.
Furthermore, notice that the total width of the histogram is
$\sum_{i=1}^n w_i$, and the overall area beneath it is $\opt$.

\medskip \noindent {\bf A collection of truncated solutions as histograms.}
We now define a collection of $m$ histograms, each of which
corresponds to a solution of a truncated input instance with
respect to some step $t$ of the algorithm. Informally, a truncated
instance is a relaxation of the input instance that admits better
solutions than the optimal solution. As a result, the collection
of histograms establishes a connection between the optimal
solution and the solution of the algorithm. For the purpose of
defining the truncated input instance corresponding to step $t$,
let $S = \{ \pi(1), \ldots, \pi(t - 1) \}$ be the set of elements
selected by the algorithm before step $t$. The truncated instance
is obtained by incrementally applying the following two
modification steps to the underlying input instance:

\medskip {\sf (i) a set of elements is given for free.} We modify
the instance by giving all the elements in $S$ for free. The
impact of this modification is two-fold: first, all the functions
that were covered by the algorithm up to that step cannot incur
any cost, and basically, they can be removed from the modified
instance; second, the threshold of each function $f^i$ that was
not covered by the algorithm up to that step decreases by
$f^i(S)$, that is, its threshold in the modified instance becomes
$\lambda_i = 1 - f^i(S)$. Now, notice that in order to translate
this instance to the canonical form in which all thresholds are
equal to $1$, one has to normalize each $f^i$ with the
corresponding threshold $\lambda_i$. It is important to note that
the marginal values of each function $f^i$ at step $t$ of the
algorithm are normalized exactly by this term (to obtain the
corresponding potential values).

\smallskip {\sf (ii) the cost of each function is relaxed.}
It is implicit in our problem definition that each function $f^i$
has a matching cost function $\rho_i: [0,1] \rightarrow \bbR_{+}$,
which represents the cost that $f^i$ collects per step with
respect to its cover. Specifically, letting $x$ indicate the
``extent of cover'' of function $f^i$, its cost is a simple
step-function, defined as follows:
$$
\rho_i(x) =
\begin{cases}
w_i & x < 1, \\
0 & x = 1.
\end{cases}
$$
Namely, the function $f^i$ collects a cost of $w_i$ in each step
until it is covered. We modify the cost function of each $f^i$ to
be continuously decreasing with constant derivative. In
particular, the updated cost function of $f^i$ becomes $\rho_i(x)
= w_i \cdot (1 - x)$. Notice that this modification implies that
even a partial cover of a function decreases its cost per step.
The interpretation one should have in mind is that $x$ represents
the fraction of covered weight, and once some fraction of weight
is covered it stops incurring cost.

\medskip The crucial observation one should make regarding the
resulting truncated instance is that the linear ordering
constructed by the optimal algorithm induces a solution for this
instance whose matching histogram is non-decreasing and completely
contained within the optimal solution histogram. The latter
argument is formally presented and proved in the following lemma.

\begin{lemma} \label{lemma:Truncated}
The optimal linear ordering constructed with respect to the
original instance induces a solution for the truncated instance
whose matching histogram is completely contained within the
optimal solution histogram when aligned to its lower right
boundary.
\end{lemma}
\begin{proof}
Prior to proving this lemma, it is important to note that the
histogram built with respect to the induced solution is defined in
a slightly different way than the optimal solution histogram. The
difference results from the fact that the cost per step of each
function $f^i$ decreases as it is covered. Specifically, this
truncated solution histogram has the same interpretation of the
axes as the optimal solution histogram, and its bars are ordered
according to non-decreasing heights. However, the number of bars
depends on the underlying solution. For instance, suppose that
function $f^i$ was incrementally covered using positive portions
$\delta_1, \ldots, \delta_k$ such that $\sum_{\ell = 1}^k
\delta_\ell = 1$ in steps $t_1, \ldots, t_k$. This function gives
rise to $k$ bars in the histogram, where bar $\ell$ has a width of
$w_i \delta_\ell$ and a height of $t_\ell$. Indeed, the total area
beneath these bars is the relative cost of function $f^i$ since
the cost per step of that function is $w_i$ between steps $t_0 =
0$ and $t_1$, $w_i \cdot (1 - \delta_1)$ between steps $t_1$ and
$t_2$, and so on. Accordingly, one can validate that
$$
w_i \sum_{\ell = 1}^k  \Big(1 - \sum_{r = 1}^{\ell - 1}
\delta_r\Big) \cdot (t_\ell - t_{\ell-1}) = w_i \sum_{\ell = 1}^k
\delta_\ell t_\ell \ .
$$
Now, for the purpose of establishing the lemma, we utilize the
following simple claim, whose proof appears in
Appendix~\ref{appsubsec:TransformationProof}. The claim presents a
transformation that may be applied to non-decreasing histograms
and does not increase their upper boundary.

\begin{claim} \label{claim:Transformation}
Consider a non-decreasing histogram and suppose we modify it by
decreasing the height of some part of a bar and then updating its
$x$-axis position to maintain the non-decreasing property. The
resulting histogram is completely contained within the primary
histogram.
\end{claim}

We demonstrate that the modifications used to generate the
truncated instance can be translated into a sequence of the
above-mentioned transformation that generates the truncated
solution histogram from the optimal solution histogram:

\medskip {\sf (i) an element is given for free.} Suppose that the
element given for free was scheduled at step $\ell$ of the optimal
linear ordering. Notice that all the elements scheduled at steps
$\ell' > \ell$ according to the optimal linear ordering are
scheduled at step $\ell' - 1$ in the induced solution. This
follows as the element under consideration does not appear in the
induced linear ordering for the truncated instance. This implies
that the cover time of all the functions that are critical with
respect to these elements decreases by $1$. We say that an element
$j$ is critical for a function if that function topped its
threshold after element $j$ was selected. Accordingly, the height
of the corresponding bars in the histogram decreases by $1$. This
translates to a sequence of the mentioned transformation.
Furthermore, notice that all the functions that were covered up to
step $\ell$ according to the optimal linear ordering may be
covered in prior steps in the induced solution. This is due to the
``free partial cover'' that the element induces. In particular,
the functions that were covered at step $\ell$ of the optimal
linear ordering must be covered in prior steps in the induced
solution. The cover time of each of these function may vary
depending on the extent of their cover with respect to the element
under consideration and previously scheduled elements. Still, it
is clear that this cover time must be strictly smaller than in the
optimal linear ordering. Consequently, the height of the bars
associated with these functions decreases. Again, this translates
to a sequence of the mentioned transformation.
Figure~\ref{fig:Translation} provides an illustration of this
modification.

\begin{figure}[!hbt]
\centerline{ \scalebox{0.35}{ \psfig{file=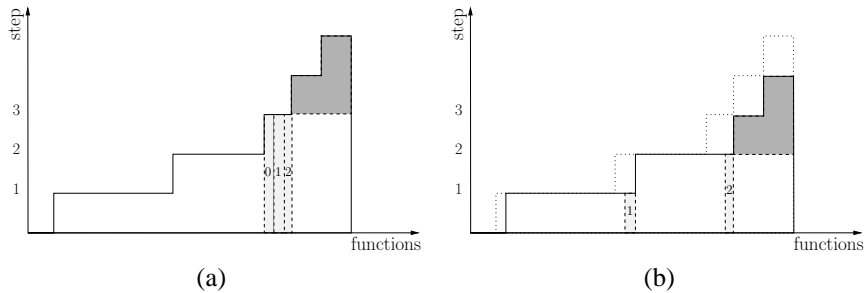} } }
\caption{The impact of giving the element scheduled at step $3$
for free: (a) The initial histogram. Note that the numbers written
inside the bars represent the new cover time of the corresponding
functions after the element under consideration was given for
free. (b) The resulting histogram.} \label{fig:Translation}
\end{figure}

\smallskip {\sf (ii) The cost of a function is relaxed.}
Suppose the cost of some function $f^i$ was relaxed, and let us
assume that $f^i$ was incrementally covered using positive
portions $\delta_1, \ldots, \delta_k$ such that $\sum_{\ell = 1}^k
\delta_\ell = 1$ in steps $t_1, \ldots, t_k$ of the optimal linear
ordering. Notice that as a result of the relaxation, the histogram
should consist of $k$ bars instead of the single bar corresponding
to function $f^i$. In particular, each bar $\ell$ should have a
width of $w_i \delta_\ell$ and a height of $t_\ell$. This can be
interpreted as replacing a single bar having respective width and
height of $w_i$ and $t_k$ with $k$ bars whose total width is $w_i
\sum_{\ell = 1}^k \delta_\ell = w_i$ and each has a height of at
most $t_k$. It is easy to verify that this may translate to a
sequence of the mentioned transformation.
\end{proof}

\noindent {\bf The solution of the algorithm as a histogram.} The
histogram that corresponds to the solution generated by the
algorithm consists of $nm$ bars, one for each entry of the
potential values matrix. The width of each bar corresponding to
entry $(i,j)$, covered at step $t$ of the algorithm, is its
weighted potential value $w_i P_{ij}$, while its height is the
penalty of the corresponding step $\Lambda_{t}$. Note that an
entry $(i,j)$ is regarded as covered at step $t$ of the algorithm
if $\pi(t) = j$. The bars are ordered according to the steps in
which the corresponding entries were covered by the algorithm.
Notice that the ordering of an element $j$ at step $t$ gives rise
to $n$ bars in the histogram whose total width is $Q_t =
\sum_{i=1}^{n} w_i P_{ij}$. Moreover, note that the total width of
the histogram is $\sum_{t=1}^m Q_t$, which is at least as large as
(and maybe much larger than) the total width of the optimal
histogram, and that the area beneath the histogram is
$\sum_{t=1}^m \Lambda_{t} Q_t$, which is precisely $\alg$, as
previously noted.

\medskip Having all the histograms definitions in place,
we are now ready to prove the theorem. We claim that the area
beneath the histogram corresponding to the solution of the
algorithm is at most $4\gamma$ times larger than the area beneath
the histogram of the optimal solution, where $\gamma = \ln(1 /
\epsilon) + 2$. Let us consider the transformation that shrinks
the width and height of each bar of the algorithm's histogram by a
factor of $2\gamma$ and $2$, respectively. Specifically, after
applying the transformation, the bar corresponding to entry
$(i,j)$ covered at step $t$ has a width of $w_i P_{ij} /
(2\gamma)$ and a height of $\Lambda_{t} / 2$. We next argue that
this shrunk histogram is completely contained within the optimal
solution histogram when aligned with its lower right boundary.
Notice that this implies that the area beneath the shrunk
histogram is no more than the area beneath the optimal solution
histogram, implying that $\alg / (4\gamma) \leq \opt$, and
therefore, proving the theorem.

For the purpose of establishing this containment argument, let us
focus on an arbitrary point $p'$ in the histogram of the
algorithm. We assume without loss of generality that it lies in
the bar corresponding to entry $(i',j')$ covered during step $t'$
of the algorithm. This implies that the height of $p'$ is at most
$\Lambda_{t'}$, and its distance from the right side boundary is
no more than $\Delta_{t'}$. Let us consider the point $p$, which
is the mapping of $p'$ in the shrunk histogram. Note that the
height of $p$ is at most $\Lambda_{t'} / 2$, while its distance
from the right boundary is at most $\Delta_{t'} / (2\gamma)$. In
the following, we prove that $p$ lies within the truncated
solution histogram corresponding to step $t'$. This is achieved by
demonstrating that the linear ordering induced by the optimal
solution for the truncated instance has at least $\Delta_{t'} /
(2\gamma)$ weight to cover by time step $\floor{\Lambda_{t'} /
2}$. In fact, we establish a more powerful property that states
that \emph{any} ordering of elements of the truncated instance has
at least $\Delta_{t'} / (2\gamma)$ weight to cover by time step
$\floor{\Lambda_{t'} / 2}$. Now, recall that
Lemma~\ref{lemma:Truncated} guarantees that the histogram of the
truncated solution is completely contained within the optimal
solution histogram, and hence, we obtain that $p$ must also lie
within the histogram of the optimal solution.

Let us concentrate on the set of elements in the truncated
instance corresponding to step $t'$. We argue that any element
selected in any step of any linear ordering cannot reduce the
weight by more than $Q_{t'} = \sum_{i = 1}^n w_i P_{ij'}$. This
argument results from the construction of the truncated instance,
the submodularity of the functions, and the greedy selection rule
of the algorithm. Specifically, the construction of the truncated
instance guarantees that the (initial) marginal value of each
element $j$ for each function $f^i$ is equal to the potential
value $P_{ij}$ at step $t'$ of the algorithm. This relates to the
normalization after the set of elements was given for free. Also
note that $P_{ij}$ can be interpreted as the fraction of the
weight $w_i$ that may be covered when selecting element $j$. This
corresponds to the modification of the cost function of each
function in the truncated instance to be continuously decreasing
with constant derivative. The submodularity of the functions,
which involve decreasing marginal values, ensures that the
marginal value of each element can only decrease over time, that
is, it cannot be greater than $P_{ij}$ in any future step. This
implies, in conjunction with the greedy selection rule of the
algorithm, which selects the element $j'$ that maximizes the
above-mentioned term, that any element selected in any step of any
linear ordering cannot cover a weight of more than $Q_{t'}$.
Consequently, any linear ordering may cover at most
$\floor{\Lambda_{t'} / 2} \cdot Q_{t'} \leq R_{t'} / 2$ weight by
step $\floor{\Lambda_{t'} / 2}$. Recall that the overall weight of
the functions in the truncated instance is exactly $R_{t'}$, and
thus, at least $R_{t'} / 2$ weight is left uncovered. Finally,
Lemma~\ref{lemma:Width} guarantees that $R_{t'} / 2 \geq
\Delta_{t'} / (2\gamma)$.~
\end{proof}

\section{An Inapproximability Result} \label{sec:Hardness}

In this section, we establish that ranking with submodular
valuations is $\Omega(\ln(1 / \epsilon))$-hard to approximate,
assuming that $\mathrm{P} \neq \mathrm{NP}$. This implies that the
outcome of the algorithm from Section~\ref{sec:AdaptiveScheme} is
optimal up to constant factors. The essence of the proof is by
showing that our problem incorporates the set cover problem as a
special instance. In fact, we demonstrate that even the seemingly
simple scenario in which there is only one function to cover
already generalizes the set cover problem.

\begin{theorem} \label{thm:Inapprox}
The ranking with submodular valuations problem cannot be
approximated within a factor of $c \ln(1 / \epsilon)$, for some
constant $c > 0$, unless $\mathrm{P} = \mathrm{NP}$.
\end{theorem}
\begin{proof}
An instance of a set cover problem consists of the ground set $X =
\{1,\ldots,n\}$ and a collection of sets ${\cal F} = \{ S_1,
\ldots, S_m \} \subseteq 2^X$. The objective is to find a
subfamily of ${\cal F}$ of minimum cardinality that covers all
elements in $X$. Set cover is known to be NP-hard to approximate
within a factor of $O(\ln n)$. In other words, there is a constant
$c > 0$ such that approximating set cover in polynomial time
within a factor of $c \ln n$ implies ${\rm P} = {\rm NP}$. This
result follows by plugging the proof system of Raz and
Safra~\cite{RazS97}, or alternatively, Arora and
Sudan~\cite{AroraS03} into a reduction of Bellare et
al.~\cite{BellareGLR93} (see also the result of
Feige~\cite{Feige98}, which shows inapproximability under a
slightly stronger assumption).

Given a set cover instance, we define an instance of ranking with
submodular valuations as follows. There are $m$ elements, each
corresponds to a set in ${\cal F}$. Furthermore, there is one
submodular set function whose corresponding weight is $1$, namely,
$w = \langle 1 \rangle$. The submodular function is defined as
$$
f(T) = \frac{1}{n} \Big|\bigcup_{\ell \in T} S_\ell\Big| \ .
$$
Notice that the resulting instance is a valid ranking with
submodular valuations instance as the function $f$ is normalized,
monotone, and submodular. For example, $f$ satisfies the
decreasing marginal values property since
$$
f_T(j) = \frac{1}{n} \Big|\Big\{ i : i \in S_j \text{ and } i
\notin \bigcup_{\ell \in T} S_\ell\Big\}\Big|
$$
is non-increasing function of $T$ for every fixed $j$.

It is easy to see that a set cover in the original instance can be
converted to a linear ordering of the elements in the
newly-created instance of identical cost. Specifically, one should
order the elements that correspond to the sets of the set cover
first (in some arbitrary way), and then order the remaining
elements. Conversely, it is not difficult to verify that given a
linear ordering of the elements, we can perform a similar
cost-preserving transformation in the opposite direction. This
implies that unless ${\rm P} = {\rm NP}$, it is impossible to
approximate the ranking with submodular valuations problem to
within a factor of $c \ln n \geq c \ln (1/\epsilon)$, where the
inequality holds as $\epsilon$ is the smallest non-zero marginal
value which is clearly at least $1 / n$.~
\end{proof}

\section{A Concluding Remark} \label{sec:Remarks}

{\bf Incorporating cost into the problem.} As previously noted, it
is implicit in the definition of the problem that each function
that needs to be covered has a matching step-function representing
the cost that the function collects per step with respect to its
cover. Specifically, the step-function corresponding to function
$f^i$ is determined by its step height $w_i$. It is only natural
to consider the generalization in which each $f^i$ has an
arbitrary non-increasing cost function $\rho_i: [0,1] \rightarrow
\bbR_{+}$ instead of the height parameter $w_i$. One can
demonstrate that our techniques can be utilized to solve this
variant. The main idea is to reduce the non-increasing cost
function case to the step-function case. This can be done by
carefully approximating each non-increasing cost function by a
collection of step-functions, as schematically described in
Figure~\ref{fig:CostfuncGeneral}. The full version of the paper
will provide a detailed description of this result.

\begin{figure}[!hbt]
\centerline{ \scalebox{0.35}{ \psfig{file=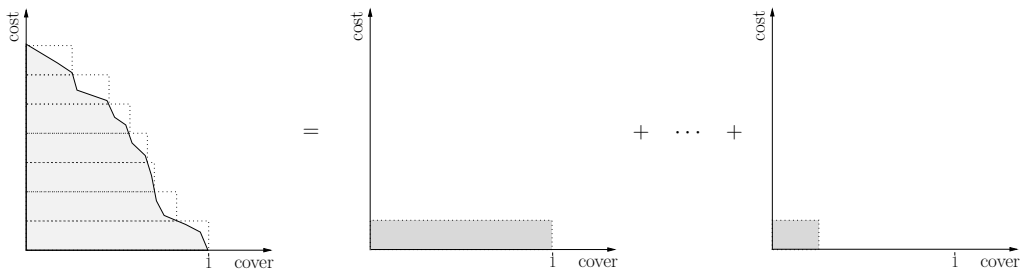} } }
\caption{The reduction from an arbitrary non-increasing cost
function to step-functions.} \label{fig:CostfuncGeneral}
\end{figure}

\paragraph{Acknowledgments:}
The authors would like to thank Oded Regev for useful discussions
on topics related to this paper.


\appendix
\section{Additional Details} \label{appsec:AdditionalDetails}
In this section, we present details omitted from the main part of
the paper.

\subsection{The natural greedy algorithm is insufficient} \label{appsubsec:NonAdaptiveFail}

Let us consider the greedy algorithm that is formally described
below. This algorithm is built from a sequence of greedy steps
that set up the linear ordering. In each step, the algorithm
selects a non-selected element that has a maximal marginal
contribution to the functions. The contribution of element $j$ to
each function $f^i$ is the gain that $j$ provides towards the
threshold of $f^i$, as formally exhibited in
line~\ref{alg:PotentialDef}.

\begin{algorithm}
\caption{Cumulative Greedy}\label{cap:CumulativeGreedy}
\begin{algorithmic}[1]
\Require A collection of $n$ submodular set functions $f^1, \ldots, f^n$, where each $f^i: 2^{[m]} \to \bbR_+$ %
\Statex \qquad\quad A weight vector $w \in \bbR_+^n$ %
\Ensure A linear ordering $\pi: [m] \to [m]$ of the ground set elements \smallskip %

\State $S \leftarrow \emptyset$ %
\For{$t \leftarrow 1$ to $m$} %
    \For{$i \leftarrow 1$ to $n$} %
        \ForAll{$j \in [m] \setminus S$} %
            \State $P_{ij} \leftarrow \begin{cases} 0 & f^i(S) \geq 1, \\ \min\left\{f^i(S \cup \{ j \}) - f^i(S), 1 - f^i(S)\right\} & \makebox{otherwise} \end{cases}$ \label{alg:PotentialDef} %
        \EndFor %
    \EndFor %
    \State Let $j \in [m] \setminus S$ be the element with maximal $\sum_{i = 1}^{n}w_i P_{ij}$ %
    \State $S \leftarrow S \cup \{ j \}$ %
    \State $\pi(t) \leftarrow j$ %
\EndFor %
\end{algorithmic}
\end{algorithm}

Unfortunately, as the following theorem demonstrates, the greedy
algorithm fails to provide good approximation. The shortfall of
the algorithm is that it misjudges elements with low marginal
contribution as less important, and thus, schedules them late. As
one may expect, this turns to be crucial when many functions
depend on a single element which has a low marginal contribution.

\begin{theorem} \label{thm:NonAdaptiveFail}
The cumulative greedy algorithm has an approximation ratio of
$\Omega(n^{1/2})$.
\end{theorem}
\begin{proof}
We consider an input instance that consists of $m = n^{1/2} + 2$
elements, $n$ functions, and a weight vector that all its entries
are identical. All the functions in the input instance are assumed
to be linear. A function $f: 2^{[m]} \to \bbR_+$ is called linear
if there is a valuations vector $v \in \bbR_+^m$ such that $f(S) =
\sum_{j \in S} v_j$. Accordingly, we represent the functions using
the following matrix.
$$
F = \left(
    \begin{array}{cc}
    \begin{array}{cc}
    1 - \frac{1}{n} & \frac{1}{n} \\
    \vdots&\vdots\\
    1 - \frac{1}{n} & \frac{1}{n} \\
    \end{array} &
    \begin{array}{ccc}
    \\
    & 0 & \\
    \\
    \end{array} \\
    \begin{array}{cc}
    \\
    & 0 \\
    \\
    \end{array} &
    \begin{array}{ccc}
    1 & 0 & 0 \\
    0 & \ddots & 0 \\
    0 & 0 & 1 \\
    \end{array}
    \end{array}
    \right) \ .
$$
In particular, the $i$-th row represents the values that function
$f^i$ has for the elements. Note that the size of $F$'s upper-left
non-zero sub-matrix is $(n - n^{1/2}) \times 2$, while the size of
its lower-right identity sub-matrix is $n^{1/2} \times n^{1/2}$.
Let us analyze the performance of the greedy algorithm on this
instance. Notice that in each step, the algorithm extends the
linear ordering with a non-selected element (column) whose sum of
entries is maximal. This follows since all weights are identical,
and the sum of entries of each row of $F$ is exactly $1$.
Consequently, the algorithm initially orders element $1$, then
elements $3$ to $n^{1/2}+2$, and finally, element $2$. The cost of
the algorithm is $(n - n^{1/2}) \cdot (n^{1/2} + 2) + (2 + \ldots
+ n^{1/2} + 1) = \Omega(n^{3/2})$. On the other hand, ordering the
elements according to their column number induces a linear
ordering whose cost is $(n - n^{1/2}) \cdot 2 + (3 + \ldots +
n^{1/2} + 2) = O(n)$.~
\end{proof}

\subsection{Proof of Claim~\ref{claim:Transformation}} \label{appsubsec:TransformationProof}
Consider an arbitrary point $p$ initially positioned at coordinate
$(x,y)$ in the histogram. Notice that unless this point is in the
area removed from the histogram due to the height decrease, it is
transformed to a new position $(x',y')$ as result of the $x$-axis
position update described in the claim. We argue that this new
position is contained in the primary histogram. If $p$ was
initially positioned in the bar whose height was decreased then
the argument clearly holds since the corresponding bar is shifted
while maintaining the non-decreasing property. Otherwise, the new
position must satisfy $x' \geq x$ and $y' = y$. However, this
implies that it must be contained in the primary histogram since
$p$ was initially in the histogram and this histogram is
non-decreasing. Figure~\ref{fig:Transformation} provides a
schematic description of the transformation.
\begin{figure}[!hbt]
\centerline{ \scalebox{0.32}{ \psfig{file=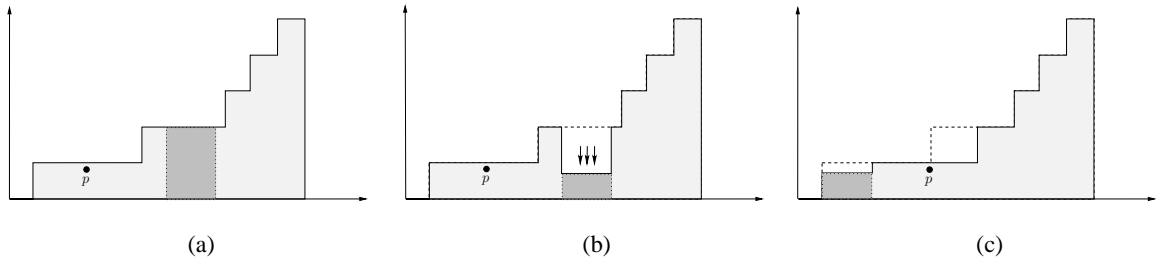} } }
\caption{The transformation described in the claim: (a) The
primary histogram; (b) The result of the height decrease; (c) The
result of the $x$-axis position update.}
\label{fig:Transformation}
\end{figure}
\end{document}